\documentclass[11pt]{article}

\usepackage[margin=1in]{geometry}
\usepackage{microtype}
\usepackage{setspace}
\usepackage{graphicx}
\usepackage{xcolor}
\usepackage{booktabs}
\usepackage{colortbl}
\usepackage{hyperref}
\hypersetup{colorlinks=true,allcolors=blue}

\usepackage{amsmath,amssymb,amsthm}
\usepackage{mathtools}

\usepackage{caption}
\usepackage{subcaption}
\usepackage{tikz}
\usepackage{algorithm}
\usepackage{algorithmic} 

\usepackage{authblk}

\title{\textbf{Certified Scalable Enclosures for Uncertain Underdetermined Systems}}

\author[1]{Rudra Prakash\thanks{Corresponding author: \href{mailto:rudra.prakash@ee.iitd.ac.in}{rudra.prakash@ee.iitd.ac.in}}}
\author[1]{Shaunak Sen}

\affil[1]{Department of Electrical Engineering, Indian Institute of Technology Delhi, New Delhi--110016, India}

\date{} 

\newcommand{\iv}[1]{\ensuremath{\left[#1\right]}} 
\newcommand{\notation}{\paragraph{Notation.}}

\newcommand{\sep}{,\ }
\newenvironment{keyword}{\par\medskip\noindent\textbf{Keywords: }\ignorespaces}{\par\medskip}

\newtheorem{theorem}{Theorem}
\newtheorem{proposition}{Proposition}
\newtheorem{corollary}{Corollary}
\newtheorem{assumption}{Assumption}
\newtheorem{definition}{Definition}
\newtheorem{remark}{Remark}
\newtheorem{mainproblem}{Problem}
\newtheorem{relatedproblem}{Related problem}
\newtheorem{example}{Example}

\begin{document}
\maketitle
\onehalfspacing

\begin{abstract}
The central challenge in underdetermined models with bounded uncertainty, such as in nonlinear design and estimation, is certifying the solution sets. Conventional solution methodologies such as those based on Newton's method or on sampling-based uncertainty quantification are either not applicable due to the underdetermined nature or do not give guarantees that all solutions have been found, assuming they converge. We addressed this issue for a problem that arises in nonlinear design, where a target steady-state box is prescribed and the parameters have to be found. We developed a trapezoidal linearisation method that rigorously encloses all solutions and combined it with a tractable linear programming method to compute the component-wise bounds. We showed that the resulting families of linear programs can iteratively contract an initial parameter region. A subdivision of the target state and the initial parameter region can improve the contraction and provide tighter enclosures. The trapezoidal relaxation gives a family of linear programs that can also be used for finding solutions for interval linear systems via the Oettli--Prager characterisation. We demonstrate the generality of the method through different applications, including nonlinear design in biomolecular circuits, a sensitivity analysis method, and in a compressed sensing context, and discuss the trade-off between enclosure tightness and computational cost.

\end{abstract}
\section{Introduction}
While multiple methods exist to analyse mathematical models, unavoidable uncertainties in the model severely limit the reliability of the obtained solutions, especially when multiple solutions are possible. This issue is prominent in systems and synthetic biology~\cite{alon_introduction_2006,del_vecchio_biomolecular_2015,sontag_new_2004,prakash_design_2024,chorasiya_quantitative_2023}, where, for example, multiple steady-state solutions correspond to different cell types. Biological parameters are usually known only in a range of values, but still need to be chosen to match the output specification. The problem is challenging because the output measurements may themselves be not known exactly. More generally, the issue of reliability with uncertain models and measurements also arises in sensing~\cite{donoho_compressed_2006} and estimation~\cite{garulli_setmembership_1999,milanese_setmembership_2004,scott_constrained_zonotopes_2016}.

There are at least three broad classes of approaches that are used to analyse models. The first class of approaches is point-based, such as those based on the Newton--Raphson method, Gradient Descent algorithms, or generalised inverses~\cite{nocedal2006numerical,ben2003generalized}. Typically, these approaches are operated as iterations, or dynamical systems, even when the equation being solved is a static one. While these approaches can compute candidate solutions, assuming they converge, they do not usually provide guarantees in scenarios with multiple solutions and/or when uncertainties are present. The problem is more pronounced in underdetermined settings when there are more unknown variables than constraint equations.
The second class of approaches is based on sampling an underlying probability distribution, such as in Bayesian inference~\cite{stuart2010inverse} or in Monte Carlo method~\cite{robert1999monte}. While these methods compute solutions and quantify uncertainties statistically, they can be sensitive to modelling choices such as priors and noise assumptions~\cite{patel_assessment_2020,ninness_mcmc_sysid_2010,schon_nonlinear_ss_sysid_2011,pillonetto_kernel_survey_2014}. Further, any nonlinear terms present pose an obstacle in rigorously propagating the probability distributions.
The third class of approaches is based on set-based computations, with Interval Analysis being a particularly attractive framework for these~\cite{moore_introduction_2009,jaulin_applied_2001,tucker2011validated}. Under a bounded set formulation, such methods can offer deterministic outer guarantees by propagating interval bounds and constructing certified enclosures, but most scalable verified solvers are developed for square determined systems~\cite{alamo_zonotopes_2005,briat2016interval,zhang2026distributed}. They often require certain regularity assumptions and may not extend directly to underdetermined settings, motivating dedicated approaches for reliably resolving underdetermined nonlinear equations~\cite{jaulin2026new}. Practical performance can also be limited by interval dependency and weak contraction in strongly coupled models, motivating subdivision and tighter relaxations~\cite{hansen2003global}. For interval linear systems, exact solution characterisation is available via the Oettli--Prager theorem, but available enclosures can be conservative and may not yield scalable component-wise bounds. For nonlinear steady-state design, constraint-propagation methods can contract boxes, but they may stagnate under dependency effects~\cite{prakash_design_2024,jaulin_applied_2001}.
An approach that combines rigour with scalability and generalisability for such uncertain underdetermined system is, to the best of our knowledge, unavailable.

In this paper, we developed a scalable, Linear Programming driven framework to obtain certified enclosures in underdetermined models with bounded, deterministic uncertainty. These algorithms output deterministic outer bounds that provably contain the entire admissible solution set. Our contributions are threefold:
\begin{itemize}
    \item We constructed tight interval-affine inclusions using a novel trapezoidal linearisation method that transforms a possibly nonlinear, uncertain, underdetermined problem into a set of linear inequalities, yielding a reusable linear program for systematic box contraction and subdivision.
    \item We derived a tractable family of linear programs to compute rigorous, component-wise outer enclosures using systematic box contraction and subdivision.
    \item We demonstrated the algorithms for a range of applications including parameter certification for nonlinear steady-state design in biomolecular systems, sensitivity analysis, and robust compressed sensing with bounded measurements.
\end{itemize}
The rest of the paper is structured as follows. Section~\ref{sec:problem_formulation} introduces the enclosure problems. Section~\ref{sec:background} reviews the main tools from interval analysis and recalls the classical parallel linearization idea. Section~\ref{sec:results} develops the main results used to compute certified enclosures, and Section~\ref{sec:applications} illustrates their performance through numerical experiments. Section~\ref{sec:discussion} discusses limitations and scalability, and Section~\ref{sec:conclusions} concludes.
\notation The notation and interval-analysis conventions used throughout are adapted from~\cite{moore_introduction_2009,jaulin_applied_2001}; biological modeling conventions follow~\cite{alon_introduction_2006,del_vecchio_biomolecular_2015}. We denote the set of all closed intervals by $\mathbb{IR}$. Interval-valued quantities are written with closed brackets, e.g., $\iv{x}\in\mathbb{IR}$, while point values are written without brackets.
An interval may also be written explicitly as $\iv{x}=[\underline{x},\overline{x}]$. An $n$-dimensional interval vector (\emph{box}) is $\iv{x}\in\mathbb{IR}^n$; interval matrices are defined analogously, e.g., $\iv{A}\in\mathbb{IR}^{m\times n}$ and $\iv{b}\in\mathbb{IR}^{m}$. For $\iv{x}=[\underline{x},\overline{x}]$, define the width $\operatorname{wid}(\iv{x}):=\overline{x}-\underline{x}$, midpoint $\operatorname{mid}(\iv{x}):=(\underline{x}+\overline{x})/2$, and radius $\operatorname{rad}(\iv{x}):=(\overline{x}-\underline{x})/2$, extended entrywise to interval vectors and matrices.
Throughout, $n$ denotes the number of unknowns and $m$ denotes the number of constraints; ``underdetermined'' means $m<n$. The symbol $x$ denotes states and $u$ denotes parameters; $\iv{x}, \iv{u}, \iv{z}\in\mathbb{IR}^{\cdot}$ denote interval boxes. For linear systems, $(A,b)$ denote point-valued realizations with $(A,b)\in(\iv{A},\iv{b})$, and $\Sigma(\iv{A},\iv{b})$ denotes the induced solution set. For the nonlinear design problem, $\iv{x}_d$ denotes the desired steady-state box and $\iv{u}_0$ the parameter search box, with $\iv{u}_{\mathrm{feas}}(\iv{x}_d,\iv{u}_0)$ the associated feasible-parameter set. In compressed sensing, $s$ denotes the unknown signal, $\iv{\Theta}$ the interval sensing matrix, and $\iv{y}$ the interval measurement vector. The symbol $w$ denotes a generic linear-program decision vector.

\section{Problem formulation}\label{sec:problem_formulation}

This section states the main problem addressed in this paper. We also note two related problems because they can be addressed with the same solution methodology.

\subsection{Motivation}\label{subsec:nonlinear_design_pf}

An important step in nonlinear system design is to fix the steady state through a suitable choice of parameters. This is a common precursor to controller design when a linearisation around this steady state is performed.
In biomolecular systems, for example, steady states are commonly characterized by
\begin{equation}\label{eq:ss_motivation}
  f(x,u) = 0,
\end{equation}
where $x \in \mathbb{R}^n$ denotes the state, $u \in \mathbb{R}^p$ denotes the reaction parameters, and $f : \mathbb{R}^n \times \mathbb{R}^p \to \mathbb{R}^n$ collects information about chemical reactions. Determining the parameters from given steady-state specifications may be underdetermined if $n<p$. This motivates our problem:

Given a parameter search region $\iv{u}_0 \in \mathbb{IR}^p$ and a target steady-state box $\iv{x}_d \in \mathbb{IR}^n$, the design task is to certify an enclosure of all parameters compatible with at least one steady state in $\iv{x}_d$.

\begin{mainproblem}[Nonlinear design]
Fix $\iv{x}_d\in\mathbb{IR}^n$ and $\iv{u}_0\in\mathbb{IR}^p$.
Define
\begin{equation}
\begin{aligned}
  \iv{u}_{\mathrm{feas}}(\iv{x}_d,\iv{u}_0)
  &:= \bigl\{\,u\in \iv{u}_0\subset\mathbb{R}^p \;\big|\; \\
  \qquad \exists\,x\in \iv{x}_d\subset\mathbb{R}^n &\text{ such that } f(x,u)=0\,\bigr\}.
\end{aligned}
\end{equation}
Compute an interval enclosure $\iv{u}\in\mathbb{IR}^p$ such that $\iv{u}_{\mathrm{feas}}(\iv{x}_d,\iv{u}_0) \subseteq \iv{u} \subseteq \iv{u}_0$.

\end{mainproblem}
Although the design constraints are typically not linear in $x$ and $u$, we used a novel Trapezoidal Linearisation technique to recast the nonlinear constraints locally over a box as an interval-affine inclusion. This yields a linear (polyhedral) relaxation in the decision variables that iteratively contracts through successive linear programs. The word `trapezoidal’ invokes the geometric shape where  the nonlinearity is enclosed, over an interval, by two not necessarily parallel lines. The next two remarks highlight other instances where the same underdetermined, bounded uncertainty structure arises. They serve to highlight the generality of this technique.
\begin{remark}\label{rem:ils}
\textbf{(Interval linear systems).} Consider an interval (possibly underdetermined) linear system
\(
  A z = b,
\)
where $z \in \mathbb{R}^n$ is unknown, $A \in \mathbb{R}^{m \times n}$, and $b \in \mathbb{R}^m$.
The coefficients are assumed to lie in intervals $A\in \iv{A}\in\mathbb{IR}^{m \times n}$ and $b\in \iv{b}\in\mathbb{IR}^{m}$, which induces the solution set
\(
  \Sigma(\iv{A}, \iv{b}) := \{ z \in \mathbb{R}^n \mid \exists \, A \in \iv{A}, \, \exists \, b \in \iv{b} \text{ such that } Az = b \}.
\)
\begin{relatedproblem}\label{prob:linear_enclosure}
Find an interval vector $\iv{z} \in \mathbb{IR}^n$ such that $\Sigma(\iv{A},\iv{b}) \subseteq \iv{z}$.
\end{relatedproblem}
\end{remark}

\begin{remark}\label{rem:rcs}
\textbf{(Robust compressed sensing).} Let $s\in\mathbb{R}^n$ be an unknown sparse signal, and suppose that both the sensing operator and the measurements are specified only within intervals: $\iv{\Theta}\in\mathbb{IR}^{m\times n}$ with $m\ll n$ and $\iv{y}\in\mathbb{IR}^m$.
\begin{relatedproblem}\label{prob:robust_cs}
Compute a sparse estimate by solving
\(
\min_{s \in \mathbb{R}^n} \lVert s \rVert_1
\, \text{subject to} \,
\underline{y} \le \Theta s \le \overline{y}, \, \Theta \in\iv{\Theta}.
\)
\end{relatedproblem}
\end{remark}

\begin{remark}\label{rem:scope}
The certificates in this paper are set-theoretic outer guarantees on the search domain: they either exclude a candidate box as infeasible or certify that a returned box encloses the entire feasible set within the search region. These enclosures need not be minimal and may be conservative due to interval dependency.
\end{remark}

\section{Background}\label{sec:background}

This section summarizes the interval-analysis and optimization methods underlying the proposed algorithms, providing the background needed to present enclosure computation and contraction operators in a self-contained way.

\subsection{Interval Analysis}

Interval analysis represents uncertain quantities as intervals and propagates them through computations to obtain guaranteed enclosures.
Arithmetic operations are extended to intervals using a set-theoretic definition:
\begin{equation*}
\iv{a} \star \iv{b}
= \{ a \star b \mid a \in \iv{a},\; b \in \iv{b} \},
\quad \text{for } \star \in \{ +, -, \times, \div \},
\end{equation*}
where division is defined assuming $0 \notin \iv{b}$. A critical property of interval arithmetic is \emph{inclusion isotonicity}: if $\iv{a_1} \subseteq \iv{a_2}$ and $\iv{b_1} \subseteq \iv{b_2}$, then
\(
\iv{a_1} \star \iv{b_1} \subseteq \iv{a_2} \star \iv{b_2}.
\)

\begin{definition}~\cite{moore_introduction_2009}
Let $f:\mathbb{R}^n\to\mathbb{R}^m$.
An \emph{interval extension} is a map $\iv{f}:\mathbb{IR}^n\to\mathbb{IR}^m$ satisfying $f(x)\in \iv{f}(\iv{x})$ for all $x\in \iv{x}$.
The \emph{range} of $f$ over $\iv{x}$ is $R(f;\iv{x}):=\{f(x):x\in \iv{x}\}$.
If $f$ is rational (built from $+,-,\times,\div$), its \emph{natural interval extension} is obtained by replacing real variables/operations with interval variables/operations.
\end{definition}

\begin{theorem}~\cite{moore_introduction_2009}
Given a function $f$ and an inclusion isotonic interval extension $\iv{f}$ defined on $\iv{x} \in \mathbb{IR}^n$:
$(i)$ If $\iv{y} \subseteq \iv{w} \subseteq \iv{x}$, then $\iv{f}(\iv{y}) \subseteq \iv{f}(\iv{w})$.
$(ii)$ $R(f;\iv{x}) \subseteq \iv{f}(\iv{x})$.

\end{theorem}
Natural extensions guarantee that a function’s true range lies within the computed bounds, but they can be conservative due to the dependency problem (a variable appearing multiple times in an expression). For Lipschitz continuous functions, this overestimation can be made arbitrarily small by subdividing the domain~\cite{bravo_interval_invariant_2005,scott_reachable_bounds_2013}.

\begin{theorem}~\cite{moore_introduction_2009}
If $f$ is Lipschitz continuous and
\(
\iv{x} = \bigcup_{i=1}^N \iv{x}_i,
\)
then
\(
R(f;\iv{x}) \subseteq \bigcup_{i=1}^N \iv{f}(\iv{x}_i) \subseteq \iv{f}(\iv{x}).
\)
Moreover, the excess width of enclosure tends to zero as the subdivision widths $\iv{x}_i$ go to zero.
\end{theorem}

\subsection{Interval linear systems}\label{sec:ils}

Consider an interval linear system with uncertain coefficients $Az=b$, where $A\in \iv{A}\in\mathbb{IR}^{m\times n}$ and $b\in \iv{b}\in\mathbb{IR}^m$. Define the associated solution set
\(
\Sigma(\iv{A},\iv{b}) := \{ z \in \mathbb{R}^n \mid \exists \, A \in \iv{A}, \, \exists \, b \in \iv{b} \text{ such that } Az = b \}.
\)
The following theorem provides an exact membership test in terms of the midpoint and radius of the interval data.
\begin{theorem}[Oettli \& Prager~\cite{oettli_compatibility_1964}]\label{theorem:oettli_prager}
A vector $z$ belongs to $\Sigma(\iv{A},\iv{b})$ if and only if
\(\label{eq:oettli_prager}
  \bigl\lvert \operatorname{mid}(\iv{A}) z - \operatorname{mid}(\iv{b}) \bigr\rvert
  \;\leq\; \operatorname{rad}(\iv{A})\,|z| + \operatorname{rad}(\iv{b}).
\)
\end{theorem}
Inequality~\eqref{eq:oettli_prager} serves as the computational basis for the interval-linear results. In Section~\ref{subsec:lp_background}, it is reformulated as linear constraints (via sign-splitting into nonnegative variables) and used within linear programs to obtain certified componentwise bounds.

\subsection{Parallel Linearization Approach~\cite{jaulin_applied_2001}}

Parallel linearization provides a standard mechanism for converting a nonlinear steady-state constraint $f(z)=0$ on a box $\iv{z}\in\mathbb{IR}^n$ into an interval-affine inclusion. For a continuously differentiable mapping $f:\mathbb{R}^n\to\mathbb{R}^m$ (typically with $n>m$), the method brackets $f$ between two affine maps,
\begin{equation}\label{Eqn_parallel_program}
  Az + \underline{b} \le f(z) \le Az + \overline{b}.
\end{equation}
\begin{assumption}\label{ass:differentiable}
The map $f$ is continuously differentiable on the current search box $\iv{z}$, and an inclusion-isotonic interval extension of its Jacobian, $\iv{J_f}(\iv{z})$, is available.
\end{assumption}

The parallel-linearization construction is based on the standard mean-value form of a differentiable mapping.
\begin{theorem}[Mean value form]~\cite{tucker2011validated}\label{thm:mvt}
Let $g:\iv{z}\to\mathbb{R}^{m}$ be continuously differentiable on a box
$\iv{z}\subseteq\mathbb{R}^{n}$. For every $z\in\iv{z}$ and every
$i\in\{1,\dots,m\}$ there exists a point $\xi^{i}$ on the segment $[z_m,z]$
such that
\begin{equation}\label{eq:mvt}
g_i(z)=g_i(z_m)+\nabla g_i(\xi^{i})^{\top}(z-z_m).
\end{equation}
Consequently, if $\iv{J}$ encloses $\{\nabla g(\zeta):\zeta\in\iv{z}\}$, then
$g(z)\in g(z_m)+\iv{J}\,(z-z_m)$ for all $z\in\iv{z}$.
\end{theorem}
Let $z_m\in \iv{z}$ (e.g., $z_m:=\operatorname{mid}(\iv{z})$). By Theorem~\ref{thm:mvt}, for each $z\in \iv{z}$ there exists $\xi\in \iv{z}$ such that
\(
  f(z)
  = f(z_m) + J_f(\xi)(z-z_m) 
  = J_f(z_m)z + f(z_m) - J_f(z_m)z_m + \bigl(J_f(\xi)-J_f(z_m)\bigr)(z-z_m).
\)
Equivalently, $f(z)=Az+b$ with
\(\label{eqn:pla_matrixA}
  A = J_f(z_m),
\)
\(\label{eqn:pla_matrixb}
  b = f(z_m) - Az_m + \bigl(J_f(\xi)-A\bigr)(z-z_m).
\)
An external approximation of $f(z)=0$ on $\iv{z}$ is therefore given by
\begin{equation}
\left\{
\begin{aligned}
& Az + b = 0,\,
A = J_f(z_m),\,
b = f(z_m) - Az_m + \\ &\bigl(J_f(\xi)-A\bigr)(z-z_m),\, z\in \iv{z},\ \xi \in \iv{z}.
\end{aligned}
\right.
\end{equation}
where $\iv{b}=[\underline{b},\overline{b}]$ in~\eqref{Eqn_parallel_program} is any interval enclosure of the admissible values of $b$ for $z,\xi\in \iv{z}$.

\subsection{Linear Programming Approach~\cite{jaulin_applied_2001}}
\label{subsec:lp_background}

This subsection presents linear-program formulations for the nonlinear relaxations (as interval-affine inequalities) and for interval linear systems (via the Oettli--Prager characterization), and uses them to compute certified componentwise enclosures. In both settings, linear programming provides certified componentwise bounds for the feasible set.

Two standard modeling/regularity conditions are imposed for the remainder of the paper. First, since the state and parameter variables represent concentrations and reaction rates, they are restricted componentwise to the nonnegative orthant. This restriction yields an endpoint-based (linear-inequality) representation of interval-affine inclusions: for $z\ge 0$, the inclusion $0\in \iv{A}z+\iv{b}$ with $\iv{A}=[\underline{A},\overline{A}]$ and $\iv{b}=[\underline{b},\overline{b}]$ is equivalent to
\begin{equation}\label{eq:int_affine_lin_ineq_bg}
\underline{A}z+\underline{b}\le 0,\qquad -\overline{A}z-\overline{b}\le 0.
\end{equation}
These inequalities are stacked into a constraint matrix $Mz\le q$.

Second, to ensure that the componentwise bounding procedure is well defined, it is assumed that for each component $i\in\{1,\ldots,n\}$ the linear programs used to minimize and maximize $z_i$ over the current feasible set are either infeasible or attain finite optimal values.
Given linear constraints of the form $Mz\le q$ together with a search box $z\in \iv{z}\in\mathbb{IR}^n$, a certified axis-aligned enclosure can be computed by solving $2n$ linear programs

\begin{equation}\label{eq:lp_enclosure_template}
\underline{z}_i:=\min\ z_i,\qquad \overline{z}_i:=\max\ z_i,\qquad i=1,\ldots,n,
\end{equation}
subject to $Mz\le q$ and $z\in \iv{z}$.

\section{Results}\label{sec:results}

The interval-valued specification $\iv{x}_d$ in the nonlinear design problem necessitates a variation of existing linearisation approaches. Because the Jacobian is interval-valued, a parallel linearisation that bounds $f$ between two affine maps with the same point-valued slope is no longer applicable (see (3)). Further, the system is underdetermined, so a Newton linearisation is also not directly applicable~\cite{jaulin_applied_2001}. We therefore bracket the nonlinear function between two affine maps with possibly different slopes, forming a trapezoidal structure. This trapezoidal linearisation is used to contract the parameter search box $\iv{u}_0$ while retaining every $u$ for which there exists $x\in\iv{x}_d$ such that $f(x,u)=0$. The result is a certified outer enclosure of feasible sets defined by underdetermined constraints with interval uncertainty. The precise statement is given in Theorem~\ref{theorem:int_linearsystem} below. The theorem also shows that, for nonnegative variables, the resulting interval-affine inclusion can be solved efficiently using linear inequalities.

\begin{theorem}\label{theorem:int_linearsystem}
Let $f:\mathbb{R}^n\times\mathbb{R}^p\to\mathbb{R}^m$ be continuously differentiable on $\iv{x}_d\times\iv{u}_0$, and let $\iv{f}$ and $\iv{J_u}$ denote inclusion-isotonic interval extensions of $f$ and of $\partial f/\partial u$, respectively. Fix a target steady-state box $\iv{x}_d\in\mathbb{IR}^n$ and a parameter search box $\iv{u}_0\in\mathbb{IR}^p$. For any parameter box $\iv{u}\subseteq\iv{u}_0$, define $u_m:=\operatorname{mid}(\iv{u})$ and
\begin{align}
&\iv{J} := \iv{J_u}(\iv{x}_d,\iv{u}),\quad \iv{A} := \iv{J_u}(\iv{x}_d,u_m),\label{eq:Adef}\\
&\iv{b} := \iv{f}(\iv{x}_d,u_m)-\iv{A}u_m+\bigl(\iv{J}-\iv{A}\bigr)\bigl(\iv{u}-u_m\bigr).\label{eq:bdef}
\end{align}
(i) If $u\in\iv{u}$ and $f(x,u)=0$ for some $x\in\iv{x}_d$, then
\begin{equation}\label{eq:incl}
0\in \iv{f}(\iv{x}_d,u) \subseteq \iv{A}u+\iv{b}.
\end{equation}
(ii) Moreover, if $u\ge 0$ componentwise and $\iv{A}=[\underline{A},\overline{A}]$, $\iv{b}=[\underline{b},\overline{b}]$ are interpreted entrywise, then
\begin{equation}\label{eq:int_affine_lin_ineq}
0\in \iv{A}u+\iv{b}
\quad\Longleftrightarrow\quad
\underline{A}u+\underline{b}\le 0,\quad -\overline{A}u-\overline{b}\le 0,
\end{equation}
with all inequalities interpreted componentwise.
\end{theorem}

\begin{proof}
Let $x\in\iv{x}_d$ and $u\in\iv{u}\subseteq\iv{u}_0$ satisfy $f(x,u)=0$. Applying Theorem~\ref{thm:mvt} to the mapping $f(x,\cdot)$ on $\iv{u}$ yields points $\xi^{i}$ on the segment $[u_m,u]$ such that
$f_i(x,u)=f_i(x,u_m)+\nabla_u f_i(x,\xi^{i})^{\top}(u-u_m)$. Each gradient lies in row $i$ of $\iv{J}$, hence
\begin{equation}\label{eq:mvincl}
0=f(x,u)\in\iv{f}(\iv{x}_d,u_m)+\iv{J}\,(u-u_m).
\end{equation}
Choose $A_1\in\iv{A}$ and write $J=A_1+(J-A_1)$ for any $J\in\iv{J}$. Then $A_1(u-u_m)\in\iv{A}(u-u_m)$, $J-A_1\in\iv{J}-\iv{A}$, and $u-u_m\in\iv{u}-u_m$, which implies
\(
\iv{J}(u-u_m)\subseteq\iv{A}(u-u_m)+\bigl(\iv{J}-\iv{A}\bigr)(\iv{u}-u_m).
\)
By subdistributivity, $\iv{A}(u-u_m)\subseteq\iv{A}u-\iv{A}u_m$. Substituting these inclusions into~\eqref{eq:mvincl} together with~\eqref{eq:bdef} yields~\eqref{eq:incl}.

Let $u\ge 0$ componentwise. Under the entrywise (box) uncertainty model, each component of $Au+b$ is affine and monotone in every uncertain entry of $A$ and $b$ for fixed $u\ge 0$, and therefore its extrema are attained at the interval endpoints. Hence, $0\in\iv{A}u+\iv{b}$ holds if and only if the endpoint inequalities in~\eqref{eq:int_affine_lin_ineq} are satisfied.
\end{proof}

\begin{remark}\label{rem:infeasibility_lp}
Infeasibility of any linear program certifies that the current box $\iv{u}$ contains no solution.
\end{remark}

\begin{remark}\label{rem:tlv_iterations}
Algorithm~\ref{alg:tlv1} describes a single LP-based contraction round (Variant~1). 
In practice, the contraction round is iterated: starting from $\iv{u}_0$, each pass computes bounds $[\underline{u}_k,\overline{u}_k]$ via the LP family. Equivalently, define the LP contraction operator
\(C_{\mathrm{LP}}(\iv{u}_k,\iv{A}_k,\iv{b}_k):=[\underline{u}_k,\overline{u}_k]\),
where $\iv{A}_k,\iv{b}_k$ are the current trapezoidal interval-affine enclosure data. The update is then
\(\iv{u}_{k+1}:=C_{\mathrm{LP}}(\iv{u}_k,\iv{A}_k,\iv{b}_k)\cap\iv{u}_k\).
If the feasible-parameter set is nonempty, then this update produces a nested sequence
\(\iv{u}_0\supseteq\iv{u}_1\supseteq\cdots\supseteq\iv{u}_k\supseteq\iv{u}_{k+1}\supseteq\cdots\).
Moreover, by inclusion isotonicity of the interval extensions, recomputing the trapezoidal enclosure on the tightened box yields nested enclosures $\iv{A}_0\supseteq\iv{A}_1\supseteq\cdots$ and $\iv{b}_0\supseteq\iv{b}_1\supseteq\cdots$.
Consequently, the interval widths are nonincreasing componentwise, i.e., $\operatorname{wid}(\iv{u}_{k+1})\le \operatorname{wid}(\iv{u}_k)$, and hence $\operatorname{wid}(\iv{u}_{k+1})/\operatorname{wid}(\iv{u}_k)\le 1$ whenever the denominator is nonzero. Iteration continues until a stopping criterion is met (e.g., no further contraction or a maximum number of passes).

\end{remark}
With entrywise (box) uncertainty, rewriting $0\in\iv{A}u+\iv{b}$ as the linear inequalities in~\eqref{eq:int_affine_lin_ineq} is exact and adds no conservatism. Any looseness comes only from the interval extensions $\iv{f}$ and $\iv{J}$. The next proposition bounds the extra enclosure width they introduce, and thus the potential conservatism of the resulting LP constraints.
\begin{proposition}\label{prop:width}
Let $\iv{A}$, $\iv{b}$ and $\iv{J}$ be as in
Theorem~\ref{theorem:int_linearsystem}. Write
$\delta:=\operatorname{rad}\bigl(\iv{J_u}(\iv{x}_d,u_m)\bigr)\in\mathbb{R}^{m\times p}_{\ge0}$,
$r:=\operatorname{rad}(\iv{u})\in\mathbb{R}^{p}_{\ge0}$ and
$\rho:=\lVert r\rVert_{\infty}$, and suppose $J_u$ is Lipschitz in its second argument, so that there is $\Lambda\in\mathbb{R}^{m\times p}_{\ge0}$ with
\(
q\bigl(\iv{J}_{ij},\,\iv{A}_{ij}\bigr)\le\Lambda_{ij}\,\rho
\qquad\text{for all }i,j,
\)
where $q$ denotes the Hausdorff distance between intervals. Then

\begin{equation}\label{eq:widthbound}
\operatorname{rad}\Bigl(\bigl(\iv{J}-\iv{A}\bigr)(\iv{u}-u_m)\Bigr)
\;\le\;\bigl(2\delta+2\Lambda\rho\bigr)\,r ,
\end{equation}
with the matrix bound read entrywise and the right-hand side a matrix--vector
product.
\end{proposition}

\begin{proof}
For intervals \(\iv{a},\iv{c}\),
\(
\operatorname{mag}(\iv{a}-\iv{c})
=\bigl|\operatorname{mid}\iv{a}-\operatorname{mid}\iv{c}\bigr|
+\operatorname{rad}\iv{a}+\operatorname{rad}\iv{c}.
\)
For \([m_1\pm r_1]\) and \([m_2\pm r_2]\), the Hausdorff distance is \(q=|m_1-m_2|+|r_1-r_2|\). Hence the Lipschitz bound yields entrywise
\(
\bigl|\operatorname{mid}\iv{J}-\operatorname{mid}\iv{A}\bigr|\le \Lambda\rho,\quad
\operatorname{rad}\iv{J}\le \operatorname{rad}\iv{A}+\Lambda\rho=\delta+\Lambda\rho,
\)
so \(\operatorname{mag}(\iv{J}-\iv{A})\le 2\delta+2\Lambda\rho\). Since \(\iv{u}-u_m\) is symmetric about \(0\) with radius \(r\), each row of the product is symmetric with radius \(\sum_j \operatorname{mag}(\iv{J}-\iv{A})_{ij} r_j\), proving \eqref{eq:widthbound}.
\end{proof}

\begin{remark}\label{rem:subdivision}
The bound \eqref{eq:widthbound} splits into a term $2\delta r$ that is linear in
the parameter radius and a term $2\Lambda\rho\,r$ that is quadratic in it.
Subdividing $\iv{u}$ shrinks the quadratic term quickly but leaves $2\delta r$
unchanged, since $\delta$ measures the spread of the Jacobian over the state box
$\iv{x}_d$ alone. Contraction can therefore stagnate at a level set by $\iv{x}_d$
rather than by $\iv{u}$, as discussed in Remark~\ref{rem:tlv_iterations}.
In this regime, tightening $\iv{x}_d$ is more
effective than further subdividing $\iv{u}$. Comparing $2\delta r$ against
$2\Lambda\rho\,r$ also provides a simple heuristic for deciding which box to
split next.

\end{remark}
These considerations are incorporated by iteratively tightening LP-based bounds and applying subdivision when the interval enclosure remains conservative.

\begin{algorithm}[t]
\caption{Trapezoidal Linearisation Variant 1}
\label{alg:tlv1}
\begin{algorithmic}[1]
\REQUIRE State box $\iv{x}$, parameter box $\iv{u}$
\ENSURE Contracted parameter box $\iv{u}$ (or infeasibility certificate)
\STATE $u_m \leftarrow \operatorname{mid}(\iv{u})$
\STATE $\iv{A} \leftarrow \iv{J_u}(\iv{x}, u_m)$
\STATE $\iv{b} \leftarrow \iv{f}(\iv{x}, u_m) - \iv{A} u_m + \bigl(\iv{J_u}(\iv{x}, \iv{u}) - \iv{A}\bigr)(\iv{u} - u_m)$
\STATE Construct the linear-program constraint matrix $Mu\le q$ by stacking:
(i)~interval-affine feasibility~\eqref{eq:int_affine_lin_ineq};
(ii)~box constraints $u\ge \underline{\iv{u}}$ and $u\le \overline{\iv{u}}$; and
(iii)~nonnegativity $u\ge 0$.
\FOR{$i=1$ to $p$}
  \STATE Solve $\underline{u}_i\leftarrow\min\ u_i$ subject to $Mu\le q$.
  \STATE Solve $\overline{u}_i\leftarrow\max\ u_i$ subject to $Mu\le q$.
\ENDFOR
\STATE Return $\iv{u}\leftarrow [\underline{u},\overline{u}]\cap \iv{u}$; if any LP is infeasible, declare $\iv{u}$ infeasible.
\end{algorithmic}
\end{algorithm}
\begin{remark}
Algorithm~\ref{alg:tlv1} performs one \emph{LP contraction round} (solve $2p$ LPs). If the enclosure $(\iv{A},\iv{b})$ is loose due to dependency, contraction can stagnate; in that case, contraction is typically paired with subdivision of $\iv{u}$ (and, when appropriate, tightening of $\iv{x}$).
\end{remark}

We next specialize the same LP-bounding template to purely interval-linear systems, for which feasibility can be characterized directly via the Oettli--Prager theorem. Let $\iv{A}\in\mathbb{IR}^{m\times n}$ and $\iv{b}\in\mathbb{IR}^m$ with $m<n$, and define the solution set
\(
\Sigma(\iv{A},\iv{b}) := \{ z \in \mathbb{R}^n \mid \exists \, A \in \iv{A}, \, \exists \, b \in \iv{b} \text{ such that } Az = b \}.
\)
By the Oettli--Prager characterization (Theorem~\ref{theorem:oettli_prager}), membership $z\in\Sigma(\iv{A},\iv{b})$ is equivalent to the nonlinear inequality~\eqref{eq:oettli_prager}. This characterization is used to derive tractable linear programs on sign-fixed domains and, for sign-indefinite variables, via a positive/negative-part reformulation.
On sign-fixed boxes, the absolute values in~\eqref{eq:oettli_prager} can be removed, leading to exact linear constraints.

\begin{proposition}\label{prop:op_lp}
Let $\iv{A}\in\mathbb{IR}^{m\times n}$ and $\iv{b}\in\mathbb{IR}^m$, and define $A_c:=\operatorname{mid}(\iv{A})$, $A_\Delta:=\operatorname{rad}(\iv{A})$, $b_c:=\operatorname{mid}(\iv{b})$, and $b_\Delta:=\operatorname{rad}(\iv{b})$. Fix a search box $\iv{z}\in\mathbb{IR}^n$ with $\iv{z}\subseteq\mathbb{R}^n_{\ge 0}$. A vector $z\in \iv{z}$ belongs to $\Sigma(\iv{A},\iv{b})$ if and only if there exists $r\in\mathbb{R}^m$ such that
\begin{equation}\label{eq:op_lp_constraints}
\begin{aligned}
-r \le A_c z-b_c \le r, \quad r\ge 0,\quad
 r \le A_\Delta z + b_\Delta.
\end{aligned}
\end{equation}
\end{proposition}

\begin{proof}
For $z\in \iv{z}\subseteq\mathbb{R}^n_{\ge 0}$ we have $|z|=z$, so~\eqref{eq:oettli_prager} is equivalent to $|A_c z-b_c|\le A_\Delta z+b_\Delta$. The constraints~\eqref{eq:op_lp_constraints} enforce this inequality by introducing $r\ge 0$ as an upper bound on $|A_c z-b_c|$.
\end{proof}

When signs are not fixed, we avoid orthant splitting by lifting to positive/negative parts.
\begin{remark}\label{rem:sign_indef_affine}
The nonnegativity restriction is convenient in biomolecular models but not essential. For sign-indefinite variables, write $z=z^+-z^-$ with $z^+,z^-\ge 0$, in which case $z^++z^-\ge |z|$ componentwise (with equality only under complementarity). This lifting yields tractable \emph{outer} relaxations for both interval-linear constraints (below) and interval-affine constraints (by replacing~\eqref{eq:int_affine_lin_ineq} with~\eqref{eq:int_affine_lin_ineq_general}) without enumerating orthants.
\end{remark}

\begin{corollary}\label{cor:op_lp_general}
Let $\iv{A}\in\mathbb{IR}^{m\times n}$ and $\iv{b}\in\mathbb{IR}^m$ with mid/radius data $(A_c,A_\Delta,b_c,b_\Delta)$ as in Proposition~\ref{prop:op_lp}. Fix a (possibly sign-indefinite) search box $\iv{z}=[\underline z,\overline z]\in\mathbb{IR}^n$. Define elementwise bounds $\overline z^+:=\max\{\overline z,0\}$ and $\overline z^-:=\max\{-\underline z,0\}$. If $z\in \iv{z}$ belongs to $\Sigma(\iv{A},\iv{b})$, then there exist $r\in\mathbb{R}^m$ and $(z^+,z^-)\in\mathbb{R}^{2n}_{\ge 0}$ such that
\begin{equation}\label{eq:op_lp_constraints_general}
\begin{aligned}
-r &\le A_c(z^+-z^-)-b_c \le r, \quad r\ge 0,\\
 r &\le A_\Delta(z^++z^-) + b_\Delta,\\
 z^+-z^- &\in \iv{z},\,
 0\le z^+ \le \overline z^+, \, 0\le z^-\le \overline z^-.
\end{aligned}
\end{equation}
Consequently,~\eqref{eq:op_lp_constraints_general} defines a linear-program-relaxation whose projection onto $z$ contains $\Sigma(\iv{A},\iv{b})\cap \iv{z}$.
\end{corollary}

\begin{proof}
Let $z\in \Sigma(\iv{A},\iv{b})\cap \iv{z}$ and define $z^+:=\max\{z,0\}$ and $z^-:=\max\{-z,0\}$ (elementwise). Then $z=z^+-z^-$, $z^+,z^-\ge 0$, $z^+\le \overline z^+$ and $z^-\le \overline z^-$ since $z\in[\underline z,\overline z]$. Moreover, $|z|=z^++z^-$. Applying the Oettli--Prager characterization~\eqref{eq:oettli_prager} gives $|A_c z-b_c|\le A_\Delta(z^++z^-)+b_\Delta$. Introducing $r\ge 0$ as an upper bound on $|A_c z-b_c|$ yields~\eqref{eq:op_lp_constraints_general}.
\end{proof}

An axis-aligned enclosure is obtained by solving componentwise linear programs over the resulting linear feasibility region.

\begin{proposition}\label{prop:lp_enclosure}
Assume that, for each component $i\in\{1,\ldots,n\}$, the linear programs in~\eqref{eq:lp_enclosure_template} are either infeasible or attain finite optimal values. If the feasible set is nonempty, then every feasible $z$ satisfies $\underline{z}_i\le z_i\le \overline{z}_i$ for all $i$, so the box $[\underline{z},\overline{z}]$ is a certified enclosure of the feasible set. If any of the $2n$ linear programs is infeasible, then the feasible set is empty.
\end{proposition}

\begin{proof}
Each $\underline{z}_i$ (respectively, $\overline{z}_i$) is, by definition, the minimum (respectively, maximum) of $z_i$ over the feasible set, so all feasible points satisfy the bounds.
\end{proof}

\begin{corollary}\label{cor:op_enclosure}
Let $\iv{z}\in\mathbb{IR}^n$ be a fixed search box. Applying Proposition~\ref{prop:lp_enclosure} to (i)~the sign-fixed constraint system~\eqref{eq:op_lp_constraints} when $\iv{z}\subseteq\mathbb{R}^n_{\ge 0}$, or (ii)~the sign-indefinite reformulation based on positive/negative parts~\eqref{eq:op_lp_constraints_general} otherwise, yields a certified outer interval box enclosure of $\Sigma(\iv{A},\iv{b})\cap \iv{z}$.
\end{corollary}

For interval-affine inclusions $0\in \iv{A}z+\iv{b}$ with sign-indefinite $z$, we avoid orthant splitting by introducing $z^+,z^-\ge 0$ with $z=z^+-z^-$. Then the interval product admits the elementwise bounds
\(\underline{A}z^+-\overline{A}z^-\le Az \le \overline{A}z^+-\underline{A}z^-\) for all $A\in \iv{A}$.
Thus $0\in \iv{A}z+\iv{b}$ is implied by the linear inequalities
\begin{equation}\label{eq:int_affine_lin_ineq_general}
\underline{A}z^+-\overline{A}z^-+\underline{b}\le 0,\qquad -\overline{A}z^+ + \underline{A}z^- -\overline{b}\le 0,
\end{equation}
with the additional constraints $z^+,z^-\ge 0$ and $z^+-z^-\in \iv{z}$ for the current search box $\iv{z}$.

Returning to the nonlinear trapezoidal linearisation, we also consider a variant that encloses the Jacobian over the full parameter box.

\begin{proposition}\label{prop:full_box_jacobian}
A second trapezoidal-linearisation variant is obtained by enclosing the Jacobian over the full box rather than centering it at $u_m$. Applying Theorem~\ref{thm:mvt} entrywise yields
\(
  f(x,u)=0 \in \iv{A}u+\iv{b},
\)
where one may take $\iv{A} = \iv{J_u}(\iv{x}, \iv{u})$ and $\iv{b} = \iv{f}(\iv{x}, u_m) - \iv{J_u}(\iv{x}, \iv{u})u_m$.
\end{proposition}

\begin{proof}
Fix $x\in\iv{x}$. By Theorem~\ref{thm:mvt} applied to the mapping $f(x,\cdot)$ on $\iv{u}$, one has $f(x,u)\in f(x,u_m)+J_u(x,\iv{u})\,(u-u_m)$ for all $u\in\iv{u}$. Using inclusion isotonicity of the interval extensions yields $f(x,u_m)\in\iv{f}(\iv{x},u_m)$ and $J_u(x,\iv{u})\subseteq\iv{J_u}(\iv{x},\iv{u})$, and hence $0=f(x,u)\in \iv{f}(\iv{x},u_m)+\iv{J_u}(\iv{x},\iv{u})\,(u-u_m)$. Defining $\iv{A}:=\iv{J_u}(\iv{x},\iv{u})$ and rearranging gives $0\in \iv{A}u+\iv{b}$ with $\iv{b}:=\iv{f}(\iv{x},u_m)-\iv{A}u_m$.

\end{proof}
To obtain Trapezoidal Linearisation Variant~2, modify Algorithm~\ref{alg:tlv1} only in the enclosure-construction step: set $\iv{A}:=\iv{J_u}(\iv{x},\iv{u})$ and $\iv{b}:=\iv{f}(\iv{x},u_m)-\iv{A}u_m$ as in Proposition~\ref{prop:full_box_jacobian}, and keep the remaining LP-based componentwise bounding steps unchanged.

\begin{remark}\label{rem:tlv1_vs_tlv2}
The two variants trade tightness for simplicity. Trapezoidal Linearisation Variant~2 (with $\iv{A}=\iv{J_u}(\iv{x},\iv{u})$) is cheapest to form but can be conservative when $J_u$ varies substantially over $\iv{u}$. Trapezoidal Linearisation Variant~1 separates the midpoint Jacobian from the deviation term and is often tighter for moderate box widths, leading to stronger contraction per LP round.
\end{remark}

\section{Applications}\label{sec:applications}

We evaluate the method on three problem classes: nonlinear circuit design (positive feedback and incoherent feedforward loops), compressed sensing recovery, and sensitivity-analysis linear systems for biochemical networks. Computations use Julia 1.11.0 with IntervalArithmetic v0.20.9 and LazySets v5.1.0. All enclosures are obtained by solving the componentwise LPs of Proposition~\ref{prop:lp_enclosure}, assuming each LP is infeasible or has a finite optimum. For interval-linear constraints, Corollary~\ref{cor:op_enclosure} links Oettli--Prager certification to LP-enclosed boxes. For nonlinear models, trapezoidal linearisation (Assumption~\ref{ass:differentiable}) yields interval-affine enclosures, contracted via Algorithm~\ref{alg:tlv1} (Variant~1) or its Variant~2 modification based on Proposition~\ref{prop:full_box_jacobian}; subdivision mitigates dependency overestimation. Code is on GitHub~\cite{ExampleCodes}.

\subsection{Design of Nonlinear Biomolecular Circuits}

We apply Trapezoidal Linearisation Variants~1 and~2 to two circuit motifs to compute guaranteed parameter bounds for desired steady-state intervals.

\begin{example}[Positive feedback loop]\label{ex:pfl}
\normalfont
The positive feedback loop is a standard motif in biomolecular regulation and can exhibit multistability under strong cooperativity~\cite{alon_introduction_2006}.
\begin{equation}
    \dot{x} = \dfrac{\alpha}{1+(K/x)^h} - \gamma x.
\end{equation}
Here, $x$ denotes the state variable and $u = [\alpha,\ \gamma,\ K]$ represents the parameter vector, with a fixed Hill coefficient $h = 10$. Standard point-based approaches (e.g., grid-based parameter scans, Monte-Carlo sampling with pointwise feasibility checks, or local parameter tuning/optimisation) can identify feasible parameter instances, but they provide neither guaranteed enclosures of the feasible parameter set nor a certificate that an entire parameter region is infeasible. In multistable regimes, inferred feasibility can depend on initialization and equilibrium-branch selection, and unstable equilibria may be missed; moreover, verifying the simultaneous existence of multiple steady states in prescribed intervals typically requires computationally expensive multi-start or branch-tracking procedures and can be sensitive to numerical tolerances.

Both Trapezoidal Linearisation Variant~1 and Trapezoidal Linearisation Variant~2 yield nearly identical bounds for the desired steady-state intervals $X_s = \{X_{s1}, X_{s2}\} = \{[6.77, 7.01], [99.9, 100.1]\}$ when initialized with the parameter box $P_0 = [50, 150] \times [0.9, 1.1] \times [5, 15]$.
The contracted parameter bounds $P_c = [89.73, 110.33] \times [0.9, 1.1] \times [8.2, 9.6]$ satisfy both steady states and illustrate that, even under strong cooperativity (large Hill coefficient), the two variants yield similar results; see Figure~\ref{fig:circuit_bounds}\subref{fig:lineardesign_pfl}.
\end{example}

\begin{example}[Incoherent feedforward loop]\label{ex:iffl}
\normalfont
The two-dimensional incoherent feedforward loop (IFFL) is a commonly used regulatory motif; depending on parameters, it can exhibit adaptation and pulsed responses~\cite{del_vecchio_biomolecular_2015}. The model considered here is

\begin{equation}
    \begin{aligned}
        \dfrac{dx_1}{dt} = \alpha - \gamma x_1, \quad
        \dfrac{dx_2}{dt} = \frac{\alpha}{1 + x_1/K} - \gamma x_2,
    \end{aligned}
\end{equation}
where $x = [x_1, x_2]$ and $u = [ \alpha, \gamma, K]$ are state and parameter vectors, respectively.
We consider the steady-state intervals $\iv{x}_s = [99.9, 100.1] \times [8.99, 9.01]$ and the initial parameter intervals $\iv{p}_0 = [50, 150] \times [0.9, 1.1] \times [9, 11]$.
Trapezoidal Linearisation Variant~1 contracts $\iv{p}_0$ to $\iv{p}_c = [89.7, 110.3] \times [0.9, 1.1] \times [9.24, 10.52]$; see Figure~\ref{fig:circuit_bounds}\subref{fig:lineardesign_iffl}.
Trapezoidal Linearisation Variant~2 yields comparable bounds.
Moreover, as $N_{\mathrm{sub}}$ and $N_{\mathrm{iter}}$ increase, Trapezoidal Linearisation Variant~2 contracts $\iv{p}_0$ further, yielding \(\iv{p}_c = [89.7, 110.3] \times [0.9, 1.1] \times [9.73, 10.05]\).
These results illustrate that the proposed approach can handle parameter uncertainty in multi-dimensional biomolecular circuit models with strong nonlinear interactions and strict steady-state constraints.

\begin{figure}[h!]
    \centering
    \begin{subfigure}[t]{0.49\linewidth}
        \centering
        \refstepcounter{subfigure}\label{fig:lineardesign_pfl}
        \begin{tikzpicture}
          \node[anchor=south west, inner sep=0] (img) {\includegraphics[trim={20cm 5cm 17cm 8cm}, clip, width=\linewidth]{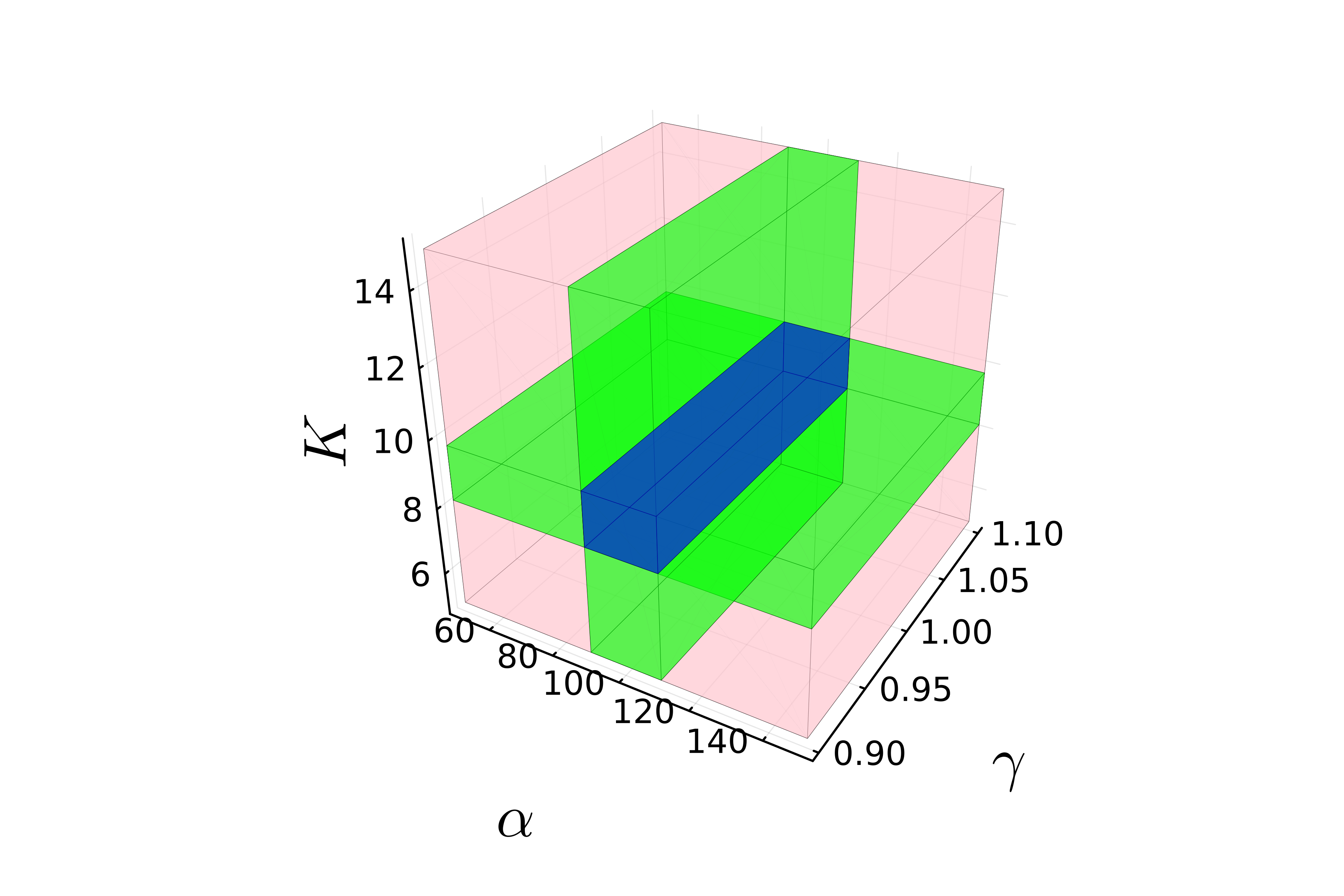}};
          \node[anchor=north west] at (img.north west) {(\alph{subfigure})};
        \end{tikzpicture}
    \end{subfigure}\hfill
    \begin{subfigure}[t]{0.49\linewidth}
        \centering
        \refstepcounter{subfigure}\label{fig:lineardesign_iffl}
        \begin{tikzpicture}
          \node[anchor=south west, inner sep=0] (img) {\includegraphics[trim={20cm 5cm 17cm 8cm}, clip, width=\linewidth]{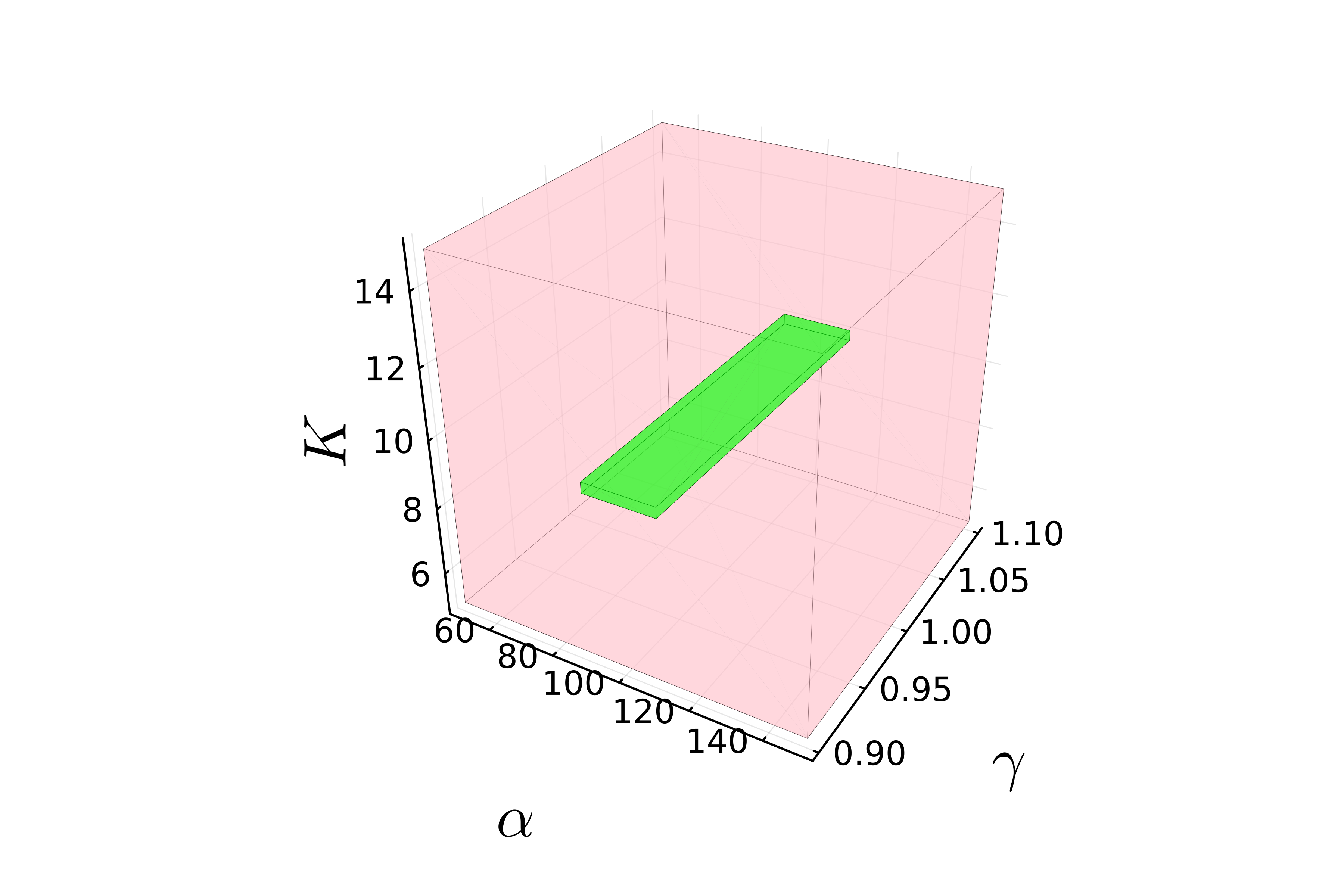}};
          \node[anchor=north west] at (img.north west) {(\alph{subfigure})};
        \end{tikzpicture}
    \end{subfigure}
    \caption{Certified parameter boxes for (a) a positive feedback loop and (b) an incoherent feedforward loop, obtained by LP contraction of the trapezoidal interval-affine enclosure using Algorithm~\ref{alg:tlv1} (Variant~1) and its Variant~2.}
    \label{fig:circuit_bounds}
\end{figure}
\end{example}

\subsection{Interval linear systems}\label{subsec:linear_apps}

The examples apply the framework to linear problems with bounded uncertainty. While interval Gauss–Seidel and Krawczyk methods are established for square, regular systems, we focus on scalable componentwise outer enclosures that also handle underdetermined cases via linear programming. Unlike nonlinear design—where parallel linearization assumes a shared point slope and fails when Jacobians vary over boxes—these linear problems admit exact Oettli–Prager characterizations expressible as linear constraints. We demonstrate the Oettli–Prager/LP enclosure method on underdetermined and square interval linear systems and on sensitivity-analysis–based interval systems.

\begin{example}[Underdetermined linear system]
\normalfont
For the interval system $\iv{A} \in \mathbb{IR}^{3 \times 4}$, the method contracts an initial search box $\iv{z}_0$ to a certified enclosure $\iv{z}_c$ of the feasible solution set.
Consider the interval matrix $\iv{A} \in \mathbb{IR}^{3\times 4}$ and interval vector $\iv{b} \in \mathbb{IR}^{3}$:
\[
\iv{A} =
\begin{bmatrix}
[11.88, 12.12] & [20.29, 20.71] & [0.49, 0.51] \\
[2.97, 3.03] & [-1.01, -0.99] & [0.89, 0.91] \\
[12.87, 13.13] & [0, 0] & [6.93, 7.07] \\
[0.1, 0.1] & [0.2, 0.2] & [-0.1, -0.1]
\end{bmatrix}^\top,
\]
\(
\iv{b} =
\begin{bmatrix}
[4.5,\ 5.5], \, [18,\ 22], \, [-5.5,\ -4.5]
\end{bmatrix}.
\)
Starting from the initial search box $\iv{z}_0 = [-10, 10]^4 \in \mathbb{IR}^{4}$, the solver combines the Oettli--Prager theorem with linear programming to contract this broad region into the tighter enclosure \(\iv{z}_c = [0.8, 1.35] \times [-0.64, 4.68] \times [-1.64, -0.48] \times [-10, 10].
\)
This contraction ensures that all feasible solutions are rigorously contained within the contracted interval solution enclosure $\iv{z}_c$.
\end{example}

\begin{example}[Square linear system]
\normalfont
We considered the interval linear system~\cite{hansen_bounding_1992}.
\(
Ax = b, \, A \in \mathbb{R}^{3 \times 3},\ b \in \mathbb{R}^{3 \times 1},\, A\in\iv{A},\ b\in\iv{b}.
\)
Since both $\iv{A}$ and $\iv{b}$ consist of interval-valued entries, we obtain the explicit system
\begin{equation}
\iv{A} = \begin{bmatrix}
[0.7,\,1.3] & [-0.3,\, 0.3] & [-0.3,\, 0.3] \\
[-0.3,\, 0.3] & [0.7,\, 1.3] & [-0.3,\, 0.3] \\
[-0.3,\, 0.3] & [-0.3,\, 0.3] & [0.7,\, 1.3]
\end{bmatrix},
\end{equation}

\begin{equation}
\iv{b} = \begin{bmatrix}
[-14,\,-7],\, [9,\,12],\, [-3,\,3]
\end{bmatrix}.
\end{equation}
Applying the proposed Oettli--Prager-based approach yields the enclosure
\(
x \in \iv{x} = [-101, 17] \times [-15, 99] \times [-90, 90].
\)
\end{example}

\begin{table*}[htbp!]

\caption{Scalability of Trapezoidal Linearisation Variants~1 and~2: subdivision factor $N_{\mathrm{sub}}$, contraction passes $N_{\mathrm{iter}}$, LP count $N_{\mathrm{LP}}$, volume ratio $\mathrm{Vol}(P_c)/\mathrm{Vol}(P_0)$, and runtime.}
\label{tab:scaling}
\centering
\small
\setlength{\tabcolsep}{4pt}
\resizebox{\textwidth}{!}{%
\begin{tabular}{| l | c | c c c c c | c c c c c|}
\hline
\rowcolor{blue!5}
 &  & \multicolumn{5}{c|}{\textbf{Trapezoidal Linearisation Variant 1}} & \multicolumn{5}{c|}{\textbf{Trapezoidal Linearisation Variant 2}} \\
\rowcolor{blue!5}
\textbf{Example} & $\boldsymbol{(n,p)}$ & $\boldsymbol{N_{\mathrm{sub}}}$ & $\boldsymbol{N_{\mathrm{iter}}}$ & $\boldsymbol{N_{\mathrm{LP}}}$ & $\boldsymbol{\frac{\mathrm{Vol}(P_c)}{\mathrm{Vol}(P_0)}}$ & \textbf{Time (s)} & $\boldsymbol{N_{\mathrm{sub}}}$ & $\boldsymbol{N_{\mathrm{iter}}}$ & $\boldsymbol{N_{\mathrm{LP}}}$ & $\boldsymbol{\frac{\mathrm{Vol}(P_c)}{\mathrm{Vol}(P_0)}}$  & \textbf{Time (s)}\\
\hline
Positive Feedback Loop & (1,3) & $2$ & $1$ & $48$   & $0.103$ & $3.36$ & $2$ & $1$ & $48$   & $0.103$ & $0.051$ \\
Incoherent Feedforward & (2,3) & $5$ & $1$ & $750$ & $0.045$ & $3.76$ & $5$ & $1$ & $750$ & $0.206$ & $0.56$ \\
\hline
\end{tabular}%
}

\end{table*}
\begin{example}[Sensitivity analysis]\label{ex:sensitivity}
\normalfont
We revisit the three-species biochemical network of Example~1 in~\cite{blanchini_bdc_2019}, with states $x=[a\ b\ c]^\top$ and inputs $u=[u_a\ u_b\ u_c]^\top$.
In BDC form, the state Jacobian is
\(
J_x(x,u)=B\,\Delta_x(x,u)\,C,
\)
with fixed integer matrices $B\in\mathbb{Z}^{3\times 8}$ and $C\in\mathbb{Z}^{8\times 3}$, and diagonal $\Delta_x$ collecting reaction-derivative magnitudes.
At steady state, the state variation satisfies the (uncertain) linear system
\(
A(\Delta_x)\,\Delta x=b(\Delta_u)\,\Delta u,
\qquad A(\Delta_x)=-B\,\Delta_x\,C,\quad b(\Delta_u)=E\,\Delta_u\,F,
\)
with $(E,F,\Delta_u,\Delta u)$ as in~\cite{blanchini_bdc_2019}. An illustrative example is
\[
A(\Delta_x)=\begin{bmatrix}
\zeta+\delta+\alpha & -\eta & -\theta+\varepsilon\\
\zeta & \eta+\beta & -\theta\\
-\zeta+\delta & -\eta & \theta+\varepsilon+\gamma
\end{bmatrix},
\, b(\Delta_u)\Delta u=\begin{bmatrix}1\\0\\0\end{bmatrix}.
\]
Using the Oettli--Prager LP enclosure (Corollary~\ref{cor:op_enclosure}), the order-of-magnitude uncertainty
\(
\alpha,\beta,\gamma,\delta,\varepsilon,\zeta,\\ \eta,\theta\in[1,10]
\)
yields the certified enclosure
\(
\Delta x\in[-0.042,0.48]\times[-1.25,0.11]\times[-0.25,0.11].
\)
The same procedure applies without any monotonicity assumption on $\Delta_x$; e.g., for
\(
\alpha,\beta,\gamma,\delta,\varepsilon,\zeta,\eta,\theta\in[-4,5],
\)
we obtain
\(
\Delta x\in[-0.083,0.067]\times[-0.25,0.2]\times[-0.11,0.11].
\)
\end{example}

\begin{example}[Compressed sensing]
\normalfont
We addressed a sparse signal recovery problem with $n = 10$ state variables and $m = 5$ measurements, resulting in an underdetermined system.
The sensing matrix $\iv{\Theta} \in \mathbb{IR}^{5 \times 10}$ and measurement vector $\iv{y} \in \mathbb{IR}^5$ were subject to interval uncertainty. A subset of the sensing matrix bounds is shown below (the complete matrix is provided in the code~\cite{ExampleCodes}):

\begin{equation}
\iv{\Theta} = \begin{bmatrix}
[1.05, 1.17] & [0.82, 0.90] & \cdots & [1.87, 2.07] \\
\vdots & \vdots & \ddots & \vdots \\
[1.94, 2.14] & [0.31, 0.34] & \cdots & \dots
\end{bmatrix}.
\end{equation}
We applied a uniform error threshold of $\pm 5\%$ to the measurements, resulting in the interval vector \(\iv{y} =
\begin{bmatrix}
[1.3, 1.42],\,
[1.5, 1.66], \,
[0.5, 0.55], \,
[1.1, 1.22], \,
[1.2, 1.32]
\end{bmatrix}.
\)
The problem was formulated as
\(
\min_{s} \; \lVert s \rVert_1 \quad \text{subject to} \\ \quad \underline{y} \le \Theta  s \le \overline{y}, \, \Theta  \in \iv{\Theta}.
\)
where $\underline{y}$ and $\overline{y}$ denote the lower and upper bounds of the interval measurements, respectively. In this experiment we enforce robust feasibility via the positive/negative-part lifting $s=s^+-s^-$ together with the corresponding linear bounds on $\iv{\Theta} s$, yielding a linear program.
The solver computed the estimated sparse signal
\(
\hat{s} = [0,\; 0.5,\; 0.28,\; 0.58,\; 0,\; 0.1,\; 0,\; 0,\; 0,\; 0].
\)
The method correctly identified the non-zero components (e.g., $s_2, s_3, s_4, s_6$) while suppressing zero components, recovering the signal despite measurement noise.
\end{example}

\section{Discussion}\label{sec:discussion}

This section highlights practical takeaways on enclosure tightness and computational cost.
Across the nonlinear design, interval-linear benchmarks, and compressed sensing the proposed linear-program-based relaxations yield certified componentwise enclosures while remaining computationally tractable for the problem sizes considered.
The main bottleneck is interval overestimation: dependency effects can inflate interval-affine bounds and weaken contraction. Subdivision reduces this conservatism, but introduces combinatorial growth. For a fixed polyhedral relaxation $Mz\le q$ over $\iv{z}\in\mathbb{IR}^n$, one enclosure pass requires $2n$ linear programs; one contraction round for $\iv{u}\in\mathbb{IR}^p$ requires $2p$ linear programs. Under uniform subdivision into $N_{\mathrm{sub}}$ intervals per parameter and $N_{\mathrm{iter}}$ contraction passes per subbox, the linear program (LP) count scales as $N_{\mathrm{LP}}=2p\,N_{\mathrm{iter}}\,N_{\mathrm{sub}}^p$.
In practice, we keep $N_{\mathrm{iter}}$ small and subdivide only when contraction stagnates. Scaling to larger $p$ will benefit from adaptive subdivision, structure-exploiting relaxations (e.g., sparsity), and efficient LP implementations (warm starts, presolve, and parallelism across subboxes).

\section{Conclusions}\label{sec:conclusions}

We developed a scalable framework to obtain guaranteed outer enclosures for solutions of a nonlinear underdetermined design problem with bounded uncertainty. Uncertainty is a major challenge in underdetermined models: solution sets are typically continua, and sampling or point-based solvers do not certify solutions that are valid for all admissible realisations. Our linear program-driven framework provides a practical route to rigorous steady-state analysis and parameter certification using standard optimisation solvers. For nonlinear steady-state design, we constructed interval-affine relaxations via a novel trapezoidal linearisation technique. We iteratively contracted the interval search boxes by solving families of linear programs, with subdivision used when dependency limits contraction. For interval linear systems, we exploited the Oettli--Prager characterization to derive tractable linear programs that certify enclosures in underdetermined and square settings that were structurally related to the nonlinear design problem, such as for sensitivity analysis and robust compressed sensing.

The transformation of uncertain, possibly nonlinear, equations into reusable linear program bounding primitives is useful in applications ranging from biomolecular circuit design, sensitivity analysis, and robust compressed sensing. The method makes the tradeoff between tightness of the binding and the runtime explicit and tenable.

\section*{Acknowledgements}

The authors thank Dr. Jagannath Samatharay for his constructive feedback.

\bibliographystyle{IEEEtran}

\bibliography{autosam-2/references}

\end{document}